\documentclass[11pt,amssymb,amsfont,a4paper]{article}
\usepackage{latexsym,amssymb,amsmath,amsthm,color}
\usepackage{geometry}
\usepackage{url}
\usepackage[utf8]{inputenc}
\usepackage{authblk}
\usepackage{bbm}
\theoremstyle{definition}
\newtheorem{definition}{Definition}

\theoremstyle{plain}
\newtheorem{theorem}{Theorem}
\newtheorem{proposition}{Proposition}
\newtheorem{lemma}{Lemma}
\newtheorem{remark}{Remark}
\newtheorem{corollary}{Corollary}

\title{Rank equivalent and rank degenerate skew cyclic codes}
\author{Umberto Mart{\'i}nez-Pe\~{n}as \thanks{umberto@math.aau.dk}}
\affil{Department of Mathematical Sciences, Aalborg University, Denmark}

\begin{document}

\maketitle

\begin{abstract}
Two skew cyclic codes can be equivalent for the Hamming metric only if they have the same length, and only the zero code is degenerate. The situation is completely different for the rank metric, where lengths of codes correspond to the number of outgoing links from the source when applying the code on a network. We study rank equivalences between skew cyclic codes of different lengths and, with the aim of finding the skew cyclic code of smallest length that is rank equivalent to a given one, we define different types of length for a given skew cyclic code, relate them and compute them in most cases. We give different characterizations of rank degenerate skew cyclic codes using conventional polynomials and linearized polynomials. Some known results on the rank weight hierarchy of cyclic codes for some lengths are obtained as particular cases and extended to all lengths and to all skew cyclic codes. Finally, we prove that the smallest length of a linear code that is rank equivalent to a given skew cyclic code can be attained by a pseudo-skew cyclic code. Throughout the paper, we find new relations between linear skew cyclic codes and their Galois closures.

\textbf{Keywords:} Cyclic codes, Galois closure, linearized polynomial rings, network coding, rank degenerate, rank distance, rank equivalence, skew cyclic codes. 

\textbf{MSC:} 15B33, 94B15, 94B65. 
\end{abstract}

\section{Introduction} \label{introduction}

Codes in the rank metric have numerous applications, such as network coding \cite{rgrw, similarities, on-metrics} or cryptography \cite{ideals}. Among these codes, cyclic codes and skew cyclic codes have been considered in \cite{skewcyclic1, skewcyclic2, oggiercyclic, gabidulin, qcyclic, ideals, rootskew}, since they have simple algebraic descriptions and fast encoding and decoding algorithms. 

In the network coding model of \cite{rgrw, similarities, on-metrics}, the length of a rank-metric code corresponds to the number of outgoing links from the source or the number of packets needed to be sent by the source. Whereas it is obvious how to increase the length of a code and preserve at the same time its rank-metric properties, just by appending zeroes, it is not obvious whether a rank-metric code can be shortened (which would mean that it is degenerate) nor how. On the other hand, skew cyclic codes of smaller length have faster encoding and decoding algorithms.

In contrast with the Hamming-metric case, skew cyclic codes may be rank equivalent and have different lengths. Moreover, many non-zero skew cyclic codes are rank degenerate. 

The aim of this paper is to study rank equivalences between skew cyclic codes, focusing on equivalences that commute with the shifting operators, and study in which way skew cyclic codes can be rank degenerate.

Both problems have direct consequences on the generalized rank weights \cite{rgrw} of skew cyclic codes, which measure the information leakage by wiretapping links in the network, following the model of \cite{rgrw, similarities, on-metrics}. In particular, from our study we will obtain as particular cases the main results in \cite{oggiercyclic}, which give exact characterizations of cyclic codes with minimal generalized rank weights by means of their root sets when their length and field size are coprime. Our results have such consequences for all lengths, do not require computing roots (or may require computing roots of linearized polynomials, which can be done efficiently) and can be applied to all skew cyclic codes.

After some preliminaries in Section 2, the results in this paper are as follows: in Section 3, we define different types of length for skew cyclic codes, regarding the rank metric, and establish some inequalities between them. In Section 4, we use the polynomial description of the Galois closure of skew cyclic codes to compute most of the lengths defined in the previous section. In Section 5, we treat cyclic codes and relate their polynomial description to that of their Galois closures (by means of generator and check polynomials, idempotent generators and root sets), giving at the end several characterizations of rank degenerate cyclic codes and obtaining the results in \cite{oggiercyclic} as particular cases. In Section 6, we proceed as in the previous section, but for general skew cyclic codes, using their linearized-polynomial description. Finally in Section 7, we see that, although the linear code of minimum length that is rank equivalent to a given skew cyclic code need not be skew cyclic, it may be chosen as pseudo-skew cyclic in many cases.

\section{Definitions and preliminaries}

Fix a prime power $ q $ and positive integers $ m $ and $ n $, and let $ \mathbb{F}_{q^s} $ denote the finite field with $ q^s $ elements for a positive integer $ s $. A code $ C \subseteq \mathbb{F}_{q^m}^n $ will be called linear if it is $ \mathbb{F}_{q^m} $-linear. In general, linearity will mean $ \mathbb{F}_{q^m} $-linearity. The number $ n $ is called the length of the code $ C $.

We will denote the coordinate indices in $ \mathbb{F}_{q^m}^n $ from $ 0 $ to $ n-1 $, and consider them as integers modulo $ n $. Given a vector $ \mathbf{c} \in \mathbb{F}_{q^m}^n $, we define its rank weight \cite{gabidulin} as the dimension of the $ \mathbb{F}_q $-linear vector space generated by its components. We denote it by $ {\rm wt_R}(\mathbf{c}) $.

Define the shifting operator $ s_n : \mathbb{F}_{q^m}^n \longrightarrow \mathbb{F}_{q^m}^n $ as
$$ s_n(c_0,c_1, \ldots, c_{n-1}) = (c_{n-1}, c_0, \ldots, c_{n-2}), $$
for every $ \mathbf{c} = (c_0, c_1, \ldots, c_{n-1}) \in \mathbb{F}_{q^m}^n $. For any integer $ r \geq 0 $, define also the $ r $-th Frobenius and $ q^r $-shifting operators as $ \theta_r, \sigma_{r,n} : \mathbb{F}_{q^m}^n \longrightarrow \mathbb{F}_{q^m}^n $, respectively, where $ \theta_r $ acts by raising every component of a vector to the power $ q^r $ and $ \sigma_{r,n} = \theta_r \circ s_n $.

\begin{definition}
A code $ C \subseteq \mathbb{F}_{q^m}^n $ is cyclic if $ s_n(C) \subseteq C $, is $ q^r $-cyclic (or skew cyclic of order $ r $) if $ \sigma_{r,n}(C) \subseteq C $, and is Galois closed (over $ \mathbb{F}_q $) if $ \theta_1(C) \subseteq C $. 
\end{definition}

Observe that cyclic codes are $ q^0 $-cyclic (or $ q^m $-cyclic), that is, they are also skew cyclic. Skew cyclic codes were introduced in \cite{gabidulin} for $ r=1 $ and $ n=m $, and then independently in \cite{qcyclic} for $ r=1 $ and in \cite{skewcyclic1} for general parameters.

Denote $ [i] = q^i $, for any integer $ i \geq 0 $. Following \cite{stichtenoth}, for a given linear code $ C \subseteq \mathbb{F}_{q^m}^n $, we define its Galois closure as $ C^* = \sum_{i=0}^{m-1}C^{[i]} $, which is the smallest linear Galois closed space containing $ C $, and we also define $ C^0 = \bigcap_{i=0}^{m-1} C^{[i]} $, which is the biggest linear Galois closed space contained in $ C $. Recall from \cite[Lemma 2]{stichtenoth} that $ (C^\perp)^* = (C^0)^\perp $ and $ (C^\perp)^0 = (C^*)^\perp $.

On the other hand, if $ V \subseteq \mathbb{F}_{q^m}^n $ and $ V^\prime \subseteq \mathbb{F}_{q^m}^{n^\prime} $ are linear Galois closed spaces, we say that a map $ \phi : V \longrightarrow V^\prime $ is a rank equivalence if it is a vector space isomorphism and $ {\rm wt_R}(\phi(\mathbf{c})) = {\rm wt_R}(\mathbf{c}) $, for all $ \mathbf{c} \in V $. We say that two codes $ C $ and $ C^\prime $ are rank equivalent if there exists a rank equivalence between linear Galois closed spaces $ V $ and $ V^\prime $ that contain $ C $ and $ C^\prime $, respectively. This definition of rank equivalent linear codes was introduced in \cite[Definition 5]{similarities}.

By \cite[Theorem 5]{similarities}, rank equivalent codes not only perform exactly in the same way regarding rank error and erasure correction, but also information leakage on networks (see \cite[Remark 5]{similarities}). Moreover, rank equivalences can be easily described, as the following lemma states, which is a particular case of \cite[Theorem 5]{similarities}:

\begin{lemma} \label{rank equivalences}
A vector space isomorphism $ \phi : V \longrightarrow V^\prime $ between linear Galois closed spaces is a rank equivalence if, and only if, there exist $ \beta \in \mathbb{F}_{q^m}^* $ and an $ n \times n^\prime $ matrix $ A $ over $ \mathbb{F}_q $ that maps bijectively $ V $ to $ V^\prime $ and such that 
$$ \phi(\mathbf{c}) = \beta \mathbf{c} A, $$
for all $ \mathbf{c} \in V $. 
\end{lemma}

As in \cite[Definition 6]{similarities}, we say that a linear code $ C \subseteq \mathbb{F}_{q^m}^n $ is rank degenerate if it is rank equivalent to a linear code with smaller length (see also \cite{slides} for an alternative equivalent definition of rank degenerate codes). In network coding this means that the code $ C $ may be implemented (with the same performance) on a network with less outgoing links or where the source needs to send less packets \cite{similarities, on-metrics}.


\section{Lengths and Galois closures}

Following the model in \cite{rgrw, similarities, on-metrics}, given a linear code $ C \subseteq \mathbb{F}_{q^m}^n $, the length $ n $ represents the number of outgoing links of a network where $ C $ is implemented, or the number of packets needed to be sent from the source, whereas $ m $ represents the packet length. If $ C $ is rank equivalent to a code with length $ n^\prime \neq n $, then it may be implemented as a linear code in a network with $ n^\prime $ outgoing links and with exactly the same performance \cite{similarities}. However we may want to implement $ C $ as a skew cyclic code, and hence we need it to be rank equivalent to a skew cyclic code of length $ n^\prime $.

On the other hand, encoding and decoding of skew cyclic codes is faster if the length is smaller, and we may always increase their lengths preserving their rank-metric properties just by appending zeroes on the right of each codeword. 

This motivates the following definitions:

\begin{definition} \label{definitions}
Given a linear code $ C \subseteq \mathbb{F}_{q^m}^n $, an element $ a \in \mathbb{F}_q^* = \mathbb{F}_q \setminus \{ 0 \} $ and an integer $ r \geq 0 $, we define the following numbers:
\begin{enumerate}
\item
The rank length, $ l_R(C) $, as the minimum $ n^\prime $ such that $ C $ is rank equivalent to a linear code of length $ n^\prime $.
\item
The $ r $-th skew length, $ l_{Sk,r}(C) $, as the minimum $ n^\prime $ such that $ C $ is rank equivalent to a linear skew cyclic code of order $ r $ and length $ n^\prime $, if such a code exists. We define $ l_{Sk,r}(C) = \infty $ otherwise.
\item
The $ (a,r) $-shift length, $ l_{Sh,a,r}(C) $, as the minimum $ n^\prime $ such that $ C $ is rank equivalent to a linear code of length $ n^\prime $ by a rank equivalence $ \phi $ such that $ a (\sigma_{r,n^\prime} \circ \phi) = \phi \circ \sigma_{r,n} $, if such a code exists. We define $ l_{Sh,a,r}(C) = \infty $ otherwise.
\item
The period length, $ l_P(C) $, as the minimum integer $ 1 \leq p \leq n $ that generates the ideal modulo $ n $ defined as $ \{ p^\prime \mid c_{i+p^\prime} = c_i, \forall i, \forall (c_0, c_1, \ldots, c_{n-1}) \in C \} $, which necessarily divides $ n $.
\end{enumerate}
We also say that an integer $ 1 \leq p \leq n $ is an $ a $-period of $ C $ if $ c_{i+p} = a c_i $, for all $ i = 0,1,2 \ldots, n-1 $ and all $ (c_0, c_1, \ldots, c_{n-1}) \in C  $.
\end{definition}

\begin{remark}
In the definition of $ l_{Sh,a,r}(C) $, the rank equivalence $ \phi $ that commutes with the $ q^r $-shifting operators is supposed to be defined between linear Galois closed spaces that are cyclic, in order to make sense (see Lemma \ref{galois cyclic} below).
\end{remark}

\begin{remark} \label{remark shift length}
Assume that $ V $ and $ V^\prime $ are linear cyclic Galois closed spaces. If a rank equivalence $ \phi : V \longrightarrow V^\prime $ satisfies that $ a (\sigma_{r,n^\prime} \circ \phi) = \phi \circ \sigma_{r,n} $, for some $ r \geq 0 $ and $ a \in \mathbb{F}_q^* $, then 
$$ \phi(\sigma_{s,n}(\mathbf{c})) = \beta^{1-[s-r]} ( \beta \sigma_{r,n}(\mathbf{c}) A)^{[s-r]} $$
$$ = \beta^{1-[s-r]} (a \sigma_{r,n^\prime}(\phi(\mathbf{c})))^{[s-r]} = (\beta^{1- [s-r]} a) \sigma_{s,n^\prime}(\phi(\mathbf{c})), $$
for all $ \mathbf{c} \in V $, where $ A $ and $ \beta $ are as in Lemma \ref{rank equivalences}. Hence, $ \phi $ sends $ q^s $-cyclic codes to $ q^s $-cyclic codes, for any $ s \geq 0 $. 
\end{remark}

We have the following on the skew cyclic structure of linear Galois closed spaces:

\begin{lemma} \label{galois cyclic}
If $ V \subseteq \mathbb{F}_{q^m}^n $ is linear and Galois closed, then it is skew cyclic of some order if, and only if, it is skew cyclic of all orders. Given a linear code $ C \subseteq \mathbb{F}_{q^m}^n $, if it is skew cyclic of some order, then $ C^* $ and $ C^0 $ are skew cyclic (of all orders).
\end{lemma}
\begin{proof}
Since $ \theta_1(V) \subseteq V $, it holds that $ \theta_r(V) = V $, for all $ r \geq 0 $. Hence, if we fix two integers $ r,s \geq 0 $, we have that $ \sigma_{r,n}(V) \subseteq V $ if, and only if, $ \sigma_{s,n}(V) \subseteq V $, and the first statement follows.

For the second statement, assume that $ C $ is $ q^r $-cyclic. It holds that
$$ \sigma_{r,n}(C^*) = \sum_{i=0}^{m-1} \sigma_{r,n}(C^{[i]}) = \sum_{i=0}^{m-1} \sigma_{r,n}(C)^{[i]} \subseteq \sum_{i=0}^{m-1} C^{[i]} = C^*, $$
and similarly for $ C^0 $, and we are done.
\end{proof}

By the discussion after \cite[Lemma 10]{similarities}, it follows that $ l_R(C) $ is equal to the $ k $-th generalized rank weight of $ C $ \cite[Definition 2]{rgrw}, for $ k = \dim(C) $, which is the dimension of $ C^* $ by \cite[Corollary 17]{slides}. That is,
\begin{equation} \label{rank length last grw}
 l_R(C) = d_{R,k}(C) = \dim(C^*).
\end{equation}
We may establish now the following relations between the different types of lengths:

\begin{proposition} \label{inequalities}
For any integers $ r,s \geq 0 $, a linear $ q^s $-cyclic code $ C \subseteq \mathbb{F}_{q^m}^n $ and an element $ a \in \mathbb{F}_q^* $, it holds that
\begin{enumerate}
\item
$ l_R(C) \leq l_{Sk,s}(C) \leq l_{Sh,a,r}(C) $.
\item
$ l_{Sh,1,r}(C) \leq l_P(C) $.
\item
$ l_R(C) = l_R(C^*) $, $ l_{Sh,a,r}(C) = l_{Sh,a,r}(C^*) $ and $ l_P(C) = l_P(C^*) $.
\item
$ l_{Sk,r}(C) \geq l_{Sk,r}(C^*) = l_R(C) $.
\end{enumerate}
\end{proposition}
\begin{proof}
In item 1, the first inequality is trivial and the second one follows from Remark \ref{remark shift length}. 

To prove item 2, we see that puncturing in the first $ l_P(C) $ coordinates gives a rank equivalence from $ C^* $ to a linear Galois closed subspace of $ \mathbb{F}_{q^m}^{l_P(C)} $ that commutes with $ \sigma_{r,n} $, and the inequality follows.

We now prove item 3. First, $ l_R(C) = \dim(C^*) = \dim(C^{**}) = l_R(C^*) $ by (\ref{rank length last grw}). Now, if $ \phi $ is a rank equivalence between $ C $ and a skew cyclic code $ C^\prime $ such that $ a (\sigma_{r,n^\prime} \circ \phi) = \phi \circ \sigma_{r,n} $, then $ \phi $ preserves Galois closures, and hence $ C^* $ is rank equivalent to $ C^{\prime *} $ by $ \phi $. It follows that $ l_{Sh,a,r}(C) \geq l_{Sh,a,r}(C^*) $, being the reversed inequality obvious. On the other hand, it follows from the definitions that $ l_P(C) = l_P(C^*) $. 

Finally, item 4 is proven in the same way as the fact that $ l_{Sh,a,r}(C) \geq l_{Sh,a,r}(C^*) $. The fact that $ l_{Sk,r}(C^*) = l_R(C^*) $ follows from the definitions.
\end{proof}

\begin{corollary} \label{singleton}
For any linear skew cyclic code $ C \subseteq \mathbb{F}_{q^m}^n $ and any $ i=R $, $ (Sk,r) $, $ (Sh,a,r) $, $ P $, we have the following Singleton-type bounds:
$$ d_R(C) \leq l_i(C) - k + 1, $$
where $ d_R $ denotes the minimum rank distance.
\end{corollary}
\begin{proof}
The case $ i = R $ follows from the classical Singleton bound \cite{gabidulin} and the fact that there exists a linear code of length $ l_R(C) $ that is rank equivalent to $ C $. The rest of the bounds follow from this case and the previous proposition.
\end{proof}

\begin{remark}
If we denote by $ d_{R,r}(C) $ the $ r $-th generalized rank weight of a linear skew cyclic code $ C \subseteq \mathbb{F}_{q^m}^n $ (see \cite[Definition 2]{rgrw}), for $ 1 \leq r \leq k $, then using that $ l_R(C) = d_{R,k}(C) $, the monotonicity of generalized rank weights \cite[Lemma 4]{rgrw} and Proposition \ref{inequalities}, we obtain the following generalized Singleton-type bounds:
$$ d_{R,r}(C) \leq l_i(C) - k + r. $$
\end{remark}

\section{Using the conventional-polynomial representation of Galois closures}

It is well-known \cite[Chapter 4]{pless} that a linear cyclic code $ C \subseteq \mathbb{F}_{q^m}^n $ can be represented as an ideal $ C(x) $ in the quotient ring $ \mathbb{F}_{q^m}[x] /(x^n  - 1) $, and it has unique polynomials $ g(x), h(x) \in \mathbb{F}_{q^m}[x] $, called generator and check polynomials, respectively, such that $ g(x) $ is monic and of minimal degree among those with residue class in $ C(x) $, and $ g(x) h(x) = x^n - 1 $. Moreover, $ g(x) $ generates $ C(x) $. 

There are two more descriptions of linear cyclic codes. If $ g(x) $ and $ h(x) $ are coprime (which holds if $ q $ and $ n $ are coprime), then there exists a unique idempotent polynomial $ e(x) \in C(x) $ (that is, $ e(x)^2 = e(x) $ in $ \mathbb{F}_{q^m}[x]/(x^n-1) $) that generates $ C(x) $ \cite[Theorem 4.3.2]{pless}. 

On the other hand, for a given polynomial $ f(x) \in \mathbb{F}_{q^m}[x] $, let $ Z(f(x)) $ denote the set of its roots in its splitting field. If $ q $ and $ n $ are coprime, then we may associate $ C $ with the root set $ Z(g(x)) $. This gives a bijective correspondence between linear cyclic codes in $ \mathbb{F}_{q^m}^n $ and root sets of divisors of $ x^n-1  $ \cite[Section 4.4]{pless}.

In this section we will focus on this conventional-polynomial representation of linear cyclic codes (in contrast with the linearized-polynomial representation in the following sections), which may be used for the Galois closure of any linear skew cyclic code by Lemma \ref{galois cyclic}.

Observe that the $ r $-th Frobenius map $ \theta_r $ induces a ring automorphism $ \theta_r : \mathbb{F}_{q^m}[x] \longrightarrow \mathbb{F}_{q^m}[x] $ given by 
\begin{equation} \label{Frobenius ring auto}
\theta_r(f_0 + f_1x + \cdots + f_dx^d) = f_0^{[r]} + f_1^{[r]} x + \cdots + f_d^{[r]}x^d,
\end{equation}
for all $ f_0 + f_1x + \cdots + f_dx^d \in \mathbb{F}_{q^m}[x] $. Since $ \theta_r(x^n - 1) = x^n -1 $, it induces a ring automorphism of the quotient ring $ \mathbb{F}_{q^m}[x]/(x^n-1) $. 

Recall from \cite[Exercise 243]{pless} that, again if $ g(x) $ and $ h(x) $ are coprime, then there exists a unique linear cyclic code $ C^c $ such that $ C \oplus C^c = \mathbb{F}_{q^m}^n $, called the cyclic complementary code of $ C $. Its generator and check polynomials are $ h(x) $ and $ g(x) $, respectively, its idempotent generator is $ 1-e(x) $ and its root set is $ Z(h(x)) = Z(x^n-1) \setminus Z(g(x)) $.

We have the following expected characterizations:

\begin{lemma} \label{polynomials}
Given a linear cyclic code $ C \subseteq \mathbb{F}_{q^m}^n $ as in the beginning of this section, the following are equivalent:
\begin{enumerate}
\item
$ C $ is Galois closed.
\item
$ g(x) \in \mathbb{F}_q[x] $.
\item
$ h(x) \in \mathbb{F}_q[x] $.
\item
(If $ g(x) $ and $ h(x) $ are coprime) $ e(x) \in \mathbb{F}_q[x] $.
\item
(If $ q $ and $ n $ are coprime) $ Z(g(x))^q = Z(g(x)) $.
\item
(If $ g(x) $ and $ h(x) $ are coprime) $ C^c $ is Galois closed.
\end{enumerate}
\end{lemma}
\begin{proof}
It is enough to note the following: 

\begin{enumerate}
\item
$ \theta_1(C) $ has $ \theta_1(g(x)) $ as generator polynomial, since it generates $ \theta_1(C)(x) $ and $ \theta_1 $ preserves monic polynomials and degrees. Hence the equivalence between items 1 and 2 follows.
\item 
$ \theta_1(C) $ has $ \theta_1(h(x)) $ as check polynomial, by the fact that $ x^n-1 = \theta_1(x^n-1) = \theta_1(g(x)) \theta_1(h(x)) $ and the previous item in this proof. Hence the equivalence between items 1 and 3 follows.
\item
If $ g(x) $ and $ h(x) $ are coprime, then $ \theta_1(C) $ has $ \theta_1(e(x)) $ as idempotent generator, since $ \theta_1(e(x)) $ is again idempotent, generates $ \theta_1(C)(x) $ and the idempotent generator is unique \cite[Theorem 4.3.2]{pless}. Hence the equivalence between items 1 and 4 follows.
\item
If $ q $ and $ n $ are coprime, then $ \theta_1(C) $ corresponds to the root set $ Z(g(x))^q $, since $ Z(\theta_1(g(x))) = Z(g(x))^q $. Hence the equivalence between items 1 and 5 follows.
\end{enumerate}
Finally, the equivalence between items 1 and 6 follows from the fact that $ h(x) $ and $ g(x) $ are the generator and check polynomials of $ C^c $, respectively.
\end{proof}

We now characterize rank equivalences that commute with the $ q^r $-shifting operators in terms of generator matrices. For a matrix $ X $ over $ \mathbb{F}_{q^m} $ with $ n $ columns, we define $ \sigma_{r,n}(X) $ as the matrix such that its $ i $-th row is the $ q^r $-shifted $ i $-th row of $ X $. 

Recall from \cite[Lemma 1]{stichtenoth} that linear Galois closed spaces are those with a basis of vectors in $ \mathbb{F}_q^n $, that is, a generator matrix with coefficients in $ \mathbb{F}_q $. 

\begin{proposition} \label{char commutes}
For linear cyclic Galois closed spaces $ V \subseteq \mathbb{F}_{q^m}^n $ and $ V^\prime \subseteq \mathbb{F}_{q^m}^{n^\prime} $, and a rank equivalence $ \phi : V \longrightarrow V^\prime $, where we define $ \beta $ and $ A $ as in Lemma \ref{rank equivalences}, the following are equivalent for a given $ a \in \mathbb{F}_q^* $ and $ r \geq 0 $:
\begin{enumerate}
\item
$ a (\sigma_{r,n^\prime} \circ \phi ) = \phi \circ \sigma_{r,n} $.
\item
If $ G $ is a generator matrix of $ V $, then $ \sigma_{r,n}(G) A = a \beta^{[r]-1} \theta_r(G) s_{n^\prime}(A) $.
\end{enumerate} 
In particular, choosing $ G $ with coefficients in $ \mathbb{F}_q $, the second item reads $ s_n(G) A = a \beta^{[r]-1} G s_{n^\prime}(A) $. Therefore, if any of the previous items hold, then $ \beta^{[r]-1} = b \in \mathbb{F}_q^* $.
\end{proposition}
\begin{proof}
For any vector $ \mathbf{c} \in V $, it is a straightforward computation to verify that condition 1 is equivalent to
$$ \sigma_{r,n}(\mathbf{c}) A = a \beta^{[r]-1} \theta_r(\mathbf{c}) s_{n^\prime}(A). $$
Therefore, the equivalence between items 1 and 2 follows from the linearity of $ \phi $ and the semi-linearity of the $ q^r $-shifting operators. 
\end{proof}

\begin{remark}
If $ V = \mathbb{F}_{q^m}^n $, then item 2 means that $ s_{n^\prime}(\mathbf{a}_i) = a^{-1} \beta^{1 - [r]} \mathbf{a}_{i+1} $, where indices $ i $ are taken modulo $ n $, and $ \mathbf{a}_i $ denotes the $ i $-th row in $ A $.

On the other hand, if $ V^\prime = \mathbb{F}_{q^m}^{n^\prime} $, then item 2 means that $ s_n(\mathbf{a}^\prime_i)= a \beta^{[r] - 1} \mathbf{a}^\prime_{i+1} $, where $ \mathbf{a}^\prime_i $ is the $ i $-th row of a matrix $ A^\prime $ with $ AA^\prime = I $ (observe that $ n^\prime \leq n $ in this case).
\end{remark}

On the other hand, the check polynomials of $ V $ and $ V^\prime $ can be easily used to see whether there exists such a rank equivalence between them, which is the first main result of this section:

\begin{theorem}
Let $ V \subseteq \mathbb{F}_{q^m}^n $ and $ V^\prime \subseteq \mathbb{F}_{q^m}^{n^\prime} $ be linear cyclic Galois closed spaces with the same dimension $ k $ and check polynomials $ h(x) $ and $ h^\prime(x) $, respectively. Given $ a \in \mathbb{F}_q^* $, an integer $ r \geq 0 $ and $ \beta \in \mathbb{F}_{q^m}^* $ such that $ \beta^{[r]} = b \beta $, for some $ b \in \mathbb{F}_q^* $, the following are equivalent:
\begin{enumerate}
\item
There exists a rank equivalence $ \phi : V \longrightarrow V^\prime $ such that $ a (\sigma_{r,n^\prime} \circ \phi) = \phi \circ \sigma_{r,n} $ and $ \beta $ is as in Lemma \ref{rank equivalences}.
\item
$ (ab)^k h^\prime(x) = h(abx) $.
\item
(If $ q $ and $ n $ are coprime) $ ab Z(h^\prime(x)) = Z(h(x)) $.
\end{enumerate}
\end{theorem}
\begin{proof}
It is obvious that items 2 and 3 are equivalent. We next prove the equivalence between items 1 and 2:

We first prove that item 1 implies item 2. Let $ h(x) = h_0 + h_1x + \cdots + h_kx^k $. Assume that there exists a rank equivalence $ \phi : V \longrightarrow V^\prime $ satisfying item 1. Let $ g(x) = g_0 + g_1 x + \cdots + g_{n-k} x^{n-k} $ be the generator polynomial of $ V $, and let $ \mathbf{g} = (g_0,g_1, \ldots, g_{n-k}, 0, \ldots, 0) \in \mathbb{F}_q^n $. Define $ f(x) = f_0 + f_1 x + \cdots + f_{n-1} x^{n-1} \in \mathbb{F}_{q^m}[x] $ such that $ \mathbf{f} = (f_0,f_1, \ldots, f_{n-1}) = \phi (\mathbf{g}) $. By Lemma \ref{rank equivalences}, we have that $ f(x) = \beta \widetilde{f}(x) $, for some $ \widetilde{f}(x) \in \mathbb{F}_q[x] $. 

It holds that $ a^i \sigma_{r,n^\prime}^i(\mathbf{f}) = \phi(\sigma_{r,n}^i(\mathbf{g})) $, for $ i = 0,1,2, \ldots, k $. In polynomial representation, we have that $ \sigma_{r,n}^i(\mathbf{g}) $ corresponds to $ x^i g(x) $, and $ x^k g(x) =  \sum_{i=0}^{k-1} -h_i x^i g(x) $. On the other hand, $ \sigma_{r, n^\prime}^i(\mathbf{f}) $ corresponds to $ x^i \theta_r^i(f(x)) = x^i \beta^{[ir]}\widetilde{f}(x) = x^i b^i \beta \widetilde{f}(x) $. Hence it follows that $ x^k a^k b^k \widetilde{f}(x) =  \sum_{i=0}^{k-1} -h_i a^i b^i x^i \widetilde{f}(x) $. In other words, $ h(abx) \widetilde{f}(x) = 0 $.

On the other hand, the vectors $ \sigma_{r,n^\prime}^i(\mathbf{f}) = a^{-i} \phi(\sigma_{r,n}^i(\mathbf{g})) $, $ i = 0,1, \ldots, k-1 $, constitute a basis of $ V^\prime $, which implies that $ \widetilde{f}(x), x\widetilde{f}(x), \ldots, x^{k-1}\widetilde{f}(x) $ constitute a basis of $ V^\prime(x) $. Hence, $ \widetilde{f}(x) $ generates the ideal $ V^\prime (x) $. Since $ h(abx)\widetilde{f}(x) = 0 $, we conclude by degrees that $ h(abx) = (ab)^k h^\prime(x) $, and we are done.

Now we prove that item 2 implies item 1. Let $ \mathbf{g}^\prime = (g_0^\prime, g_1^\prime, \ldots, g_{n^\prime-k}^\prime, 0, \ldots, 0) \in \mathbb{F}_{q^m}^{n^\prime} $ by such that $ g^\prime(x) = g_0^\prime + g_1^\prime x + \cdots + g_{n^\prime-k}^\prime x^{n^\prime - k} $ is the generator polynomial of $ V^\prime $. We just need to define $ \phi $ by the formula 
\begin{equation} \label{constructing equi}
\phi(\sigma_{r,n}^i(\mathbf{g})) = a^i \sigma_{r,n^\prime}^i(\beta \mathbf{g}^\prime) = (a^ib^i)\beta \sigma_{r,n^\prime}^i(\mathbf{g}^\prime),
\end{equation}
for $ i = 0,1,2, \ldots, k-1 $, which defines a rank equivalence between $ V $ and $ V^\prime $ by Lemma \ref{rank equivalences}, and see that this formula also holds for $ i = k $. If that happens, then the equality $ a (\sigma_{r,n^\prime} \circ \phi) = \phi \circ \sigma_{r,n} $ holds on a basis of $ V $ and then it holds on all $ V $.

To see that Equation (\ref{constructing equi}) also holds for $ i=k $, we may argue as in the converse implication by using again that $ (ab)^k h^\prime(x) = h(abx) $.
\end{proof}

Observe from the previous proof that each of such equivalences is constructed by Equation (\ref{constructing equi}) using a polynomial $ f(x) \in V^\prime(x) $ with coefficients in $ \mathbb{F}_q $ and such that $ f(x), xf(x), \ldots, x^{k-1} f(x) $ generate $ V^\prime(x) $ as a vector space.

Taking $ a=b=1 $, we obtain the following particular case:

\begin{corollary}
Let $ V, V^\prime $ and $ h(x),h^\prime(x) $ be as in the previous theorem. The following are equivalent:
\begin{enumerate}
\item
There exists a matrix $ A \in \mathbb{F}_q^{n \times n^\prime} $, mapping $ V $ to $ V^\prime $, such that $ \phi : V \longrightarrow V^\prime $ given by $ \phi(\mathbf{c}) = \mathbf{c} A $, $ \mathbf{c} \in V $, satisfies $ \sigma_{r,n^\prime} \circ \phi = \phi \circ \sigma_{r,n} $.
\item
$ h^\prime(x) = h(x) $.
\item
(If $ q $ and $ n $ are coprime) $ Z(h^\prime(x)) = Z(h(x)) $.
\end{enumerate}
\end{corollary}

On the other hand, given a polynomial $ f(x) \in \mathbb{F}_{q^m}[x] $, we define its order as the minimum positive integer $ e $ such that $ f(x) $ divides $ x^e-1 $ (in $ \mathbb{F}_{q^m}[x] $), and denote it by $ {\rm ord}(f(x)) $. In general, for $ a \in \mathbb{F}_q $, we define the $ a $-order of $ f(x) $ as the minimum positive integer $ e $ such that $ f(x) $ divides $ x^e - a^e $ (in $ \mathbb{F}_{q^m}[x] $), if one such $ e $ exists, and denote it by $ {\rm ord}_a(f(x)) $. If no such $ e $ exists, we define $ {\rm ord}_a(f(x)) = \infty $.

We may now prove the second main result of this section:

\begin{theorem} \label{char different lengths}
For an integer $ s \geq 0 $ and a linear $ q^s $-cyclic code $ C \subseteq \mathbb{F}_{q^m}^n $, where $ h^0(x) $ is the check polynomial of $ C^* $, it holds that
\begin{enumerate}
\item
$ l_R(C) = \deg(h^0(x)) $.
\item
$ l_{Sh,1,0}(C) = l_P(C) = {\rm ord}(h^0(x)) \leq n $.
\item
More generally, if $ a \in \mathbb{F}_q^* $, then $ e = l_{Sh,a,0}(C) = {\rm ord}_a(h^0(x)) $ and $ e $ is an $ a^e $-period of $ C $.
\item
More generally, if $ a \in \mathbb{F}_q^* $ and $ r \geq 0 $, then 
$$ l_{Sh,a,r}(C) = \min \{ {\rm ord}_{ab}(h^0(x)) \mid b \in \mathbb{F}_q^*, \beta \in \mathbb{F}_{q^m}^*, \beta^{[r]} = b \beta \}. $$
\end{enumerate}
In particular, $ l_R(C) = l_{Sk,s}(C) = l_{Sh,1,r}(C) = l_P(C) $ if, and only if, $ \deg(h^0(x)) = {\rm ord}(h^0(x)) $, which holds if, and only if, $ h^0(x) = x^e-1 $, for some positive integer $ e $.
\end{theorem}
\begin{proof}
First of all, we have seen that $ l_R(C) = \dim(C^*) $, and this dimension is $ \deg(h^0(x)) $. Hence item 1 follows. 

Now, the equality $ l_{Sh,1,0}(C) = {\rm ord}(h^0(x)) $ in item 2 follows from item 3 by choosing $ a=1 $, and it is straightforward to see that $ {\rm ord}(h^0(x)) $ is equal to $ l_P(C) $.

Item 3 follows now from item 4, since $ \beta^{[0]} = \beta $, for all $ \beta \in \mathbb{F}_{q^m} $. Moreover, since $ x^e f(x) = a^e f(x) $, for all $ f(x) \in C(x) $, we see that $ e $ is an $ a^e $-period of $ C $.

Next we prove item 4. Assume that there exists a rank equivalence $ \phi : C^* \longrightarrow V^\prime $, where $ V^\prime $ is a linear cyclic Galois closed space of length $ e $ and $ a (\sigma_{r,n^\prime} \circ \phi) = \phi \circ \sigma_{r,n} $. By the previous theorem, the check polynomial of $ V^\prime $ is $ (ab)^{-k}h^0(abx) $, $ k = \dim(C^*) $, with notation as in the previous theorem. Hence, we see that $ h^0(x) $ divides $ x^e - (ab)^e $, since $ h^0(abx) $ divides $ x^e - 1 $. 

Conversely, if $ h^0(x) $ divides $ x^e - (ab)^e $, we may define the linear cyclic Galois closed space $ V^\prime \subseteq \mathbb{F}_{q^m}^e $ with check polynomial $ h^\prime(x) = (ab)^{-k}h^0(abx) $, which divides $ x^e -1 $. Then there exists a rank equivalence $ \phi : C^* \longrightarrow V^\prime $ as before by the previous theorem.

Therefore, choosing the elements $ b $ and $ \beta $ that minimize the number $ e $, we see that $ l_{Sh,a,r}(C) = {\rm ord}_{ab}(h^0(x)) $, and item 4 follows. 

Finally, by Proposition \ref{inequalities}, we conclude that $ l_R(C) = l_{Sk,s}(C) = l_{Sh,1,r}(C) = l_P(C) $ if, and only if, $ \deg(h^0(x)) = {\rm ord}(h^0(x)) $. It is straightforward to see that this is equivalent to $ h^0(x) = x^e-1 $, $ e = {\rm ord}(h^0(x)) $.
\end{proof}

\begin{remark} \label{attaining}
We see from the previous theorem that there are three instances of $ h^0(x) $ that give an easy description of $ C $ and where $ l_R(C) $ is attained by a linear skew cyclic code:
\begin{enumerate}
\item
$ h^0(x) = x^e - 1 $, for some positive integer $ e \leq n $. This case has two subcases:
\begin{enumerate}
\item
$ e = n $, which corresponds to the case where $ C $ is not rank degenerate.
\item
$ e < n $, in which case $ C $ is rank degenerate in a special way: $ e $ divides $ n $ and $ C $ is constructed by repeating $ n/e $ times a linear skew cyclic code $ C^\prime \subseteq \mathbb{F}_{q^m}^e $ that is not rank degenerate.
\end{enumerate}
\item
There exists an $ a \in \mathbb{F}_q^* $ and a positive integer $ e < n $ with $ h^0(x) = x^e - a^e $. In this case, $ e $ divides $ n $, $ a^n = 1 $ and again $ C $ is constructed by repeating $ n/e $ times a linear skew cyclic code $ C^\prime \subseteq \mathbb{F}_{q^m}^e $ that is not rank degenerate.

Observe that $ a^n=1 $ since $ x^e - a^e $ divides $ x^n-1 $. Then we see that $ a^e((a^{-1}x)^e - 1) $ divides $ (a^{-1}x)^n - 1 = x^n-1 $, hence $ x^e-1 $ divides $ x^n-1 $, which implies that $ e $ divides $ n $.
\end{enumerate}
\end{remark}

In the next sections we will give results on $ C $ in terms of its structure, and relate it to the structure of $ C^* $ and $ C^0 $.

\section{Cyclic codes, conventional polynomials and root sets} \label{conventional}

Given a polynomial $ f(x) \in \mathbb{F}_{q^m}[x] $, we define the following polynomials, where divisibility is considered in $ \mathbb{F}_{q^m}[x] $:
\begin{equation} \label{star}
f^*(x) = {\rm gcd}(f(x), \theta_1(f(x)), \ldots, \theta_{m-1}(f(x))),
\end{equation}
\begin{equation} \label{zero}
f^0(x) = {\rm lcm}(f(x), \theta_1(f(x)), \ldots, \theta_{m-1}(f(x))),
\end{equation}
\begin{equation} \label{dual}
f^\perp(x) = x^{\deg(f(x))}f(x^{-1}) / f(0),
\end{equation}
assuming $ f(0) \neq 0 $ in the last equation. We have the following:

\begin{lemma}
For any polynomial $ f(x) \in \mathbb{F}_{q^m}[x] $, it holds that $ f^*(x), f^0(x) \in \mathbb{F}_q[x] $.
\end{lemma}
\begin{proof}
Since $ \theta_1 $ leaves the set $ \{ f(x), \theta_1(f(x)), \ldots, \theta_{m-1}(f(x)) \} $ invariant and is a ring automorphism, it holds that $ \theta_1(f^*(x)) = f^*(x) $ and $ \theta_1(f^0(x)) = f^0(x) $, which mean that both lie in $ \mathbb{F}_q[x] $.
\end{proof}

Now fix a linear cyclic code $ C \subseteq \mathbb{F}_{q^m}^n $, whose generator and check polynomials are $ g(x) $ and $ h(x) $, respectively. It is well-known that $ h^\perp(x) $ and $ g^\perp(x) $ are the generator and check polynomials of $ C^\perp $, respectively \cite[Theorem 4.2.7]{pless}. The following proposition explains the previous notation:

\begin{proposition} \label{gen and check of closure}
The generator and check polynomials of $ C^* $ are $ g^*(x) $ and $ h^0(x) $, respectively, and the generator and check polynomials of $ C^0 $ are $ g^0(x) $ and $ h^*(x) $, respectively. In particular, 
$$ (g^*)^\perp(x) = (g^\perp)^*(x), \quad (g^0)^\perp(x) = (g^\perp)^0(x), $$
$$ (h^*)^\perp(x) = (h^\perp)^*(x), \quad \textrm{and} \quad (h^0)^\perp(x) = (h^\perp)^0(x). $$
\end{proposition}
\begin{proof}
Taking ideals, it holds that $ C(x) = \sum_{i=0}^{m-1} C^{[i]}(x) $, and $ C^{[i]} $ has $ \theta_i(g(x)) $ as generator polynomial. It is well-known that the generator polynomial of the sum of cyclic codes is the greatest common divisor of their generator polynomials \cite[Theorem 4.3.7]{pless}.

Hence the polynomial $ g^*(x) $ is the generator polynomial of $ C^* $. Similarly $ g^0(x) $ is the generator polynomial of $ C^0 $, using now that the generator polynomial of the intersection of cyclic codes is the least common multiple of their generator polynomials \cite[Theorem 4.3.7]{pless}.

On the other hand, we have that $ \theta_i(g(x)) \theta_i(h(x)) = \theta_i(x^n-1) = x^n - 1 $, for $ i = 0,1,2, \ldots, m-1 $. Hence the greatest common divisor of the polynomials $ \theta_i(g(x)) $ and the least common multiple of the polynomials $ \theta_i(h(x)) $ satisfy the same. That is, $ g^*(x) h^0(x) = x^n - 1 $, and $ h^0(x) $ is the check polynomial of $ C^* $. Similarly for $ C^0 $.

Finally, since $ g^\perp(x) $ is the check polynomial of $ C^\perp $, it follows that $ (g^\perp)^*(x) $ is the check polynomial of $ (C^\perp)^0 $. On the other hand, since $ g^*(x) $ is the generator polynomial of $ C^* $, it holds that $ (g^*)^\perp(x) $ is the check polynomial of $ (C^*)^\perp $. Since $ (C^\perp)^0 = (C^*)^\perp $, it follows that $ (g^*)^\perp(x) = (g^\perp)^*(x) $. The remaining equalities are proven in the same way.
\end{proof}

On the other hand, we have the following relations of idempotent generators and cyclic complementaries. 

\begin{proposition}
Assume that $ g(x) $ and $ h(x) $ are coprime and $ e(x) $ is the idempotent generator of $ C $. Then $ C^* $ and $ C^0 $ have $ 1- \prod_{i=0}^{m-1} (1 - \theta_i(e(x))) $ and $ \prod_{i=0}^{m-1} \theta_i(e(x)) $ as idempotent generators, respectively. Moreover it holds that
$$ (C^c)^* = (C^0)^c \quad \textrm{and} \quad (C^c)^0 = (C^*)^c. $$
\end{proposition} 
\begin{proof}
First, the idempotent generator of the intersection of cyclic codes is the product of their idempotent generators \cite[Theorem 4.3.7]{pless}, hence $ \prod_{i=0}^{m-1} \theta_i(e(x)) $ is the idempotent generator of $ C^0 $.

On the other hand, $ C^c $ has $ h(x) $ as generator polynomial, thus $ (C^c)^* $ has $ h^*(x) $ as generator polynomial by the previous proposition. Moreover, $ C^0 $ has $ h^*(x) $ as check polynomial, also by the previous proposition. Therefore $ (C^c)^* = (C^0)^c $. Similarly we may prove that $ (C^c)^0 = (C^*)^c $.

Finally, It holds that $ \prod_{i=0}^{m-1} (1 - \theta_i(e(x))) $ is the idempotent generator of $ (C^c)^0 $ by the first part of this proof. Using that $ (C^c)^0 = (C^*)^c $, we see that $ 1- \prod_{i=0}^{m-1} (1 - \theta_i(e(x))) $ is the idempotent generator of $ C^* $.
\end{proof}

In addition, we may easily see that $ C^c $ and $ C^\perp $ are rank equivalent:

\begin{proposition}
Assume that $ g(x) $ and $ h(x) $ are coprime. Then $ C^c $ and $ C^\perp $ are rank equivalent.
\end{proposition}
\begin{proof}
There exists a permutation of indices that maps $ C^c $ to $ C^\perp $ by \cite[Theorem 4.4.9]{pless}, which obviously defines a rank equivalences between them.
\end{proof}

We will now relate $ C $, $ C^* $ and $ C^0 $ by means of the defining root set of $ C $. We will relate $ l_R(C^\perp) $ with the parameter $ \eta_q(C) $ introduced in \cite{oggiercyclic}, which will allow us to easily derive the main results in that paper. 

If $ q $ and $ n $ are coprime, let $ m^\prime \geq m $ be such that $ \mathbb{F}_{q^{m^\prime}} $ is the splitting field of $ g(x) $. Let $ \alpha_1, \alpha_2, \ldots, \alpha_{n-k} \in \mathbb{F}_{q^{m^\prime}} $ be the simple roots of $ g(x) $, and assume that they are ordered in the following way: there exist $ 1 = m_0 < m_1 < m_2 < \ldots < m_t = n-k+1 $ such that $ \alpha_{m_i}, \alpha_{m_i+1}, \ldots, \alpha_{m_{i+1}-1} $ are roots of the minimal polynomial $ \mu_i(x) \in \mathbb{F}_q[x] $ of $ \alpha_{m_i} $ over $ \mathbb{F}_q $, for $ i = 0,1, \ldots, t-1 $.

\begin{definition}[\textbf{\cite[Definition 3, Definition 4]{oggiercyclic}}]
With notation as in the previous paragraph, we define 
$$ \mu_q(g(x)) = \prod_{i=0}^{t-1} \mu_r(x) \in \mathbb{F}_q[x] \quad \textrm{and} \quad \eta_q(C) = \deg(\mu_q(g(x))). $$
\end{definition} 

We have the following relations, which in particular compute the root set corresponding to $ C^* $ and $ C^0 $:

\begin{proposition} \label{coprime then}
If $ q $ and $ n $ are coprime, then 
\begin{enumerate}
\item
$ \mu_q(g(x)) = g^0(x) $.
\item
$ Z(g^0(x)) = \bigcup_{i=0}^{m-1} Z(g(x))^{[i]} $ and $ Z(g^*(x)) = \bigcap_{i=0}^{m-1} Z(g(x))^{[i]} $.
\item
$ \eta_q(C) = \deg(g^0(x)) = \dim((C^\perp)^*) $ and $ \deg(g^*(x)) = \dim((C^\perp)^0) $.
\end{enumerate}
Analogous identities hold replacing $ g(x) $, $ g^*(x) $ and $ g^0(x) $ by $ h(x) $, $ h^*(x) $ and $ h^0(x) $, respectively.
\end{proposition}
\begin{proof}
First, $ g(x) $ divides $ \mu_q(g(x)) $ in $ \mathbb{F}_{q^{m^\prime}}[x] $ by looking at their roots. By the same argument as in Lemma \ref{polynomials}, we see that $ g(x) $ divides $ \mu_q(g(x)) $ in $ \mathbb{F}_{q^m}[x] $.

Fix a positive integer $ r $. Since $ \theta_r $ is a ring isomorphism, we see that $ \theta_r(g(x)) $ also divides $ \theta_r(\mu_q(g(x))) = \mu_q(g(x)) $ in $ \mathbb{F}_{q^m}[x] $. Hence $ \mu_q(g(x)) $ is divisible by the least common multiple of the polynomials $ \theta_r(g(x)) $, $ r = 0,1,2, \ldots, m-1 $.

Finally, since the Galois group of the extension $ \mathbb{F}_q \subseteq \mathbb{F}_{q^{m^\prime}} $ is constituted by the maps $ \theta_r $, we see that the previous least common multiple vanishes at the roots of the polynomials $ \mu_i(x) $, for $ i=0,1, \ldots, t-1 $. Hence $ \mu_q(g(x)) = {\rm lcm}(g(x), \theta_1(g(x)), \ldots, \theta_{m-1}(g(x))) = g^0(x) $ and item 1 follows.

By the same discussion, since $ Z(\theta_i(g(x))) = Z(g(x))^{[i]} $, we have that $ Z(g^0(x)) = \bigcup_{i=0}^{m-1} Z(g(x))^{[i]} $. On the other hand, denoting $ Z = Z(x^n - 1) $ and using that $ g^*(x)h^0(x) = g(x)h(x) = x^n-1 $, we have that
$$ Z(g^*(x)) = Z \setminus Z(h^0(x)) = Z \setminus \left( \bigcup_{i=0}^{m-1} Z(h(x))^{[i]} \right) = \bigcap_{i=0}^{m-1} \left( Z \setminus Z(h(x)) \right) ^{[i]} = \bigcap_{i=0}^{m-1} Z(g(x))^{[i]}, $$
and item 2 follows. Item 3 follows from item 1 and Proposition \ref{gen and check of closure}.
\end{proof}

\begin{remark} \label{implying oggiercyclic}
Hence $ C^\perp $ is rank degenerate if, and only if, $ \eta_q(C) < n $, which by the duality theorem for generalized rank weights \cite[Theorem]{jerome} is equivalent to $ d_R(C) = 1 $ (see \cite{jerome} for more details). Hence \cite[Proposition 2]{oggiercyclic} and \cite[Proposition 3]{oggiercyclic} follow. We have actually proven that 
\begin{equation} \label{eta equality}
\eta_q(C) = l_R(C^\perp),
\end{equation}
which combined with the same duality theorem also implies \cite[Proposition 5]{oggiercyclic}. Moreover, together with Corollary \ref{singleton} we obtain \cite[Proposition 6]{oggiercyclic}. \\
\end{remark}

We may now state the main result of this section, which computes lengths of cyclic codes in terms of their intrinsic structure:

\begin{theorem} \label{rank degenerate cyclic}
It holds that
$$ l_R(C) = n - \deg({\rm gcd}(g(x), \theta_1(g(x)), \ldots, \theta_{m-1}(g(x)))) $$
$$ = \deg({\rm lcm}(h(x), \theta_1(h(x)), \ldots, \theta_{m-1}(h(x)))), $$
and if $ q $ and $ n $ are coprime, then
$$ l_R(C) = \eta_q(C^\perp) = n - \# \left( \bigcap_{i=0}^{m-1} Z(g(x))^{[i]} \right) = \# \left( \bigcup_{i=0}^{m-1} Z(h(x))^{[i]} \right). $$
On the other hand, for $ a \in \mathbb{F}_q^* $, it holds that
$$ l_{Sh,a,0}(C) = {\rm ord}_a(h^0(x)) = {\rm ord}_a(h(x)) = {\rm ord}_a(\mu_q(h(x))) $$
$$ = \min \{ e \mid \alpha^e = a^e, \forall \alpha \in Z(h(x)) \}. $$
\end{theorem}
\begin{proof}
The first two equalities follow from Theorem \ref{char different lengths}, item 1, and Proposition \ref{gen and check of closure}. If $ q $ and $ n $ are coprime, then the next three equalities follow from the same results as before together with Proposition \ref{coprime then} and Equation (\ref{eta equality}). Finally, the last four equalities follow from the same results as before together with Theorem \ref{char different lengths}, items 2 and 3.
\end{proof}

The following characterizations of rank degenerate cyclic codes follow:

\begin{corollary}
The following conditions are equivalent:
\begin{enumerate}
\item
$ C $ is rank degenerate. That is, $ l_R(C) < n $.
\item
$ {\rm gcd}(g(x), \theta_1(g(x)), \ldots, \theta_{m-1}(g(x))) \neq 1 $.
\item
$ {\rm lcm}(h(x), \theta_1(h(x)), \ldots, \theta_{m-1}(h(x))) \neq x^n - 1 $.
\item
(If $ g(x) $ and $ h(x) $ are coprime) $ \prod_{i=0}^{m-1}(1 - \theta_i(e(x)))) = 0 $ in $ \mathbb{F}_{q^m}[x]/(x^n-1) $, where $ e(x) $ is the idempotent generator of $ C $.
\item
(If $ q $ and $ n $ are coprime) $ \eta_q(C^\perp) < n $.
\item
(If $ q $ and $ n $ are coprime) $ \bigcap_{i=0}^{m-1} Z(g(x))^{[i]} \neq \emptyset $.
\item
(If $ q $ and $ n $ are coprime) $ \bigcup_{i=0}^{m-1} Z(h(x))^{[i]} \subsetneq Z(x^n - 1) $.
\item
$ g(x) $ is divisible by some non-constant polynomial $ f(x) \in \mathbb{F}_q[x] $ (in $ \mathbb{F}_{q^m}[x] $).
\item
$ h(x) $ divides some polynomial $ f(x) \in \mathbb{F}_q[x] $ (in $ \mathbb{F}_{q^m}[x] $) of degree less than $ n $.
\end{enumerate} 
\end{corollary}

\section{Skew cyclic codes, linearized polynomials and root spaces} \label{linearized}

In this section we will fix a positive integer $ r $ and assume that $ m $ divides $ rn $, and will use the linearized-polynomial description of skew cyclic codes given in \cite{skewcyclic1, gabidulin, qcyclic, rootskew} to give similar characterizations of lengths and rank degenerateness as in the previous section for general skew cyclic codes. By the discussion after \cite[Remark 2]{rootskew}, assuming that $ m $ divides $ rn $ does not leave any skew cyclic code out of study regarding the lengths $ l_i(C) $.

Denote by $ \mathcal{L}_{q^r} \mathbb{F}_{q^m}[x] $ the ring of $ q^r $-linearized polynomials over $ \mathbb{F}_{q^m} $ (see \cite{gabidulin, orespecial, ore} or \cite[Chapter 3]{lidl}), that is, polynomials of the form 
$$ F(x) = F_0 x + F_1 x^{[r]} + F_2 x^{[2r]} + \cdots + F_d x^{[dr]}, $$
where $ F_0, F_1, F_2, \ldots, F_d \in \mathbb{F}_{q^m}[x] $, and where we consider composition of maps $ \otimes $ as product. We also define the $ q^r $-degree of $ F(x) $ as $ \deg_{q^r}(F(x)) = d $ if $ F_d \neq 0 $. 

Recall that $ q^r $-linearized polynomials over $ \mathbb{F}_{q^m} $ define $ \mathbb{F}_{q^r} $-linear maps between field extensions of $ \mathbb{F}_{q^r} $ and their compositions as such define again $ q^r $-linearized polynomials over $ \mathbb{F}_{q^m} $. This ring constitutes an Euclidean domain on the right and on the left \cite{gabidulin, orespecial, ore}, but we will always consider divisibility on the right. We will also use the term ``conventional'' to refer to the usual product and divisibility of polynomials.

Since $ m $ divides $ rn $, $ x^{[rn]} - x $ commutes with every other $ q^r $-linearized polynomial over $ \mathbb{F}_{q^m} $ and the left ideal $ (x^{[rn]} - x) $ is two-sided. Thus, we may consider the ring $ \mathcal{L}_{q^r} \mathbb{F}_{q^m}[x] / (x^{[rn]} - x) $, which is isomorphic to $ \mathbb{F}_{q^m}^n $ as a vector space.

Linear $ q^r $-cyclic codes correspond to left ideals in $ \mathcal{L}_{q^r} \mathbb{F}_{q^m}[x] / (x^{[rn]} - x) $ \cite{skewcyclic1, gabidulin, qcyclic}. Fix one $ C \subseteq \mathbb{F}_{q^m}^n $. It has unique generator polynomial $ G(x) $ and check polynomial $ H(x) $ with the same properties as in the usual case \cite{skewcyclic1, qcyclic, rootskew}: $ G(x) $ is of minimal degree and monic, and $ x^{[rn]} - x = G(x) \otimes H(x) = H(x) \otimes G(x) $.

For a given $ F(x) \in \mathcal{L}_{q^r} \mathbb{F}_{q^m}[x] $, we will also write $ F = F(x) + (x^{[rn]} - x) $, the residue class of $ F(x) $ modulo $ x^{[rn]} - x $. Recall that, since $ q^r $-linearized polynomials induce $ \mathbb{F}_{q^r} $-linear maps, their root sets are $ \mathbb{F}_{q^r} $-linear vector spaces. We may denote by $ Z(F) $ the $ \mathbb{F}_{q^r} $-linear space of zeroes in $ \mathbb{F}_{q^{rn}} $ of $ F(x) \in \mathcal{L}_{q^r} \mathbb{F}_{q^m}[x] $. This definition is consistent, since two $ q^r $-polynomials $ F_1(x) $ and $ F_2(x) $ have the same roots in $ \mathbb{F}_{q^r} $ if $ F_1(x) - F_2(x) \in (x^{[rn]} - x) $.

On the other hand, the $ s $-th Frobenius map $ \theta_s $ defines also a ring automorphism $ \theta_s :\mathcal{L}_{q^r} \mathbb{F}_{q^m}[x] \longrightarrow \mathcal{L}_{q^r} \mathbb{F}_{q^m}[x] $ using the same formula as in the conventional case (\ref{Frobenius ring auto}) and induces a ring automorphism of $ \mathcal{L}_{q^r} \mathbb{F}_{q^m}[x] / (x^{[rn]} - x) $, since $ \theta_r(x^{[rn]} - x) = x^{[rn]} - x $. \\

In this section we will consider the $ q^r $-cyclic structure of linear Galois closed spaces. However, describing generator and check polynomials of $ C^\perp $, $ C^* $ and $ C^0 $ is not as straightforward as in the conventional case. Given a $ q^r $-polynomial $ F(x) = F_0x + F_1 x^{[r]} + \cdots + F_{d} x^{[rd]} \in \mathcal{L}_{q^r} \mathbb{F}_{q^m}[x] $ that divides $ x^{[rn]} - x $, with $ F_d \neq 0 $, we define:
\begin{equation} \label{dual lin 1}
F^\perp(x) = \left(\frac{F_d}{F_0^{[dr]}} \right) x + \left(\frac{F_{d-1}^{[r]}}{F_0^{[dr]}} \right) x^{[r]} + \cdots + \left(\frac{F_0^{[dr]}}{F_0^{[dr]}} \right) x^{[dr]},
\end{equation}
\begin{equation} \label{dual lin 2}
F^\top(x) = \left(\frac{F_d}{F_0} \right)^{[(n-d)r]} x + \left(\frac{F_{d-1}}{F_0} \right)^{[(n-d+1)r]} x^{[r]} + \cdots + \left(\frac{F_0}{F_0} \right)^{[nr]} x^{[dr]},
\end{equation}
\begin{equation} \label{star lin 1}
F^*(x) = {\rm gcd}(F(x), \theta_1(F(x)), \ldots, \theta_{m-1}(F(x))),
\end{equation}
\begin{equation} \label{zero lin 1}
F^0(x) = {\rm lcm}(F(x), \theta_1(F(x)), \ldots, \theta_{m-1}(F(x))),
\end{equation}
\begin{equation} \label{star lin 2}
F_*(x) = {\rm gcd}(F(x)^\perp, \theta_1(F(x))^\perp, \ldots, \theta_{m-1}(F(x))^\perp)^\top,
\end{equation}
\begin{equation} \label{zero lin 2}
F_0(x) = {\rm lcm}(F(x)^\perp, \theta_1(F(x))^\perp, \ldots, \theta_{m-1}(F(x))^\perp)^\top.
\end{equation}

As in the previous section, we have the following:

\begin{lemma}
For any $ q^r $-polynomial $ F(x) \in \mathcal{L}_{q^r} \mathbb{F}_{q^m}[x] $, it holds that $ F^*(x) $, $ F^0(x) $, $ F_*(x) $, $ F_0(x) \in \mathcal{L}_{q^r} \mathbb{F}_{q}[x] $. 
\end{lemma}
\begin{proof}
Since $ \theta_1 $ leaves the set $ \{ F(x), \theta_1(F(x)), \ldots, \theta_{m-1}(F(x)) \} $ invariant and is a ring automorphism, it holds that $ \theta_1(F^*(x)) = F^*(x) $ and $ \theta_1(F^0(x)) = F^0(x) $. Observing that $ \theta_1(F^\perp(x)) = \theta_1(F(x))^\perp $ and $ \theta_1(F^\top(x)) = \theta_1(F(x))^\top $, we see that $ \theta_1(F_*(x)) = F_*(x) $ and $ \theta_1(F_0(x)) = F_0(x) $. Hence the result follows.
\end{proof}

\begin{remark}
Observe that taking $ r = m $, we obtain a $ q^m $-linearized description of cyclic codes and $ n $ may be arbitrary. Observe that $ \mathcal{L}_{q^m} \mathbb{F}_{q^m}[x] $ is commutative and naturally isomorphic to $ \mathbb{F}_{q^m}[x] $ by the map given in \cite[Definition 3.58]{lidl}:
\begin{equation} \label{conventional equi to linearized}
L(f_0 + f_1x + \cdots + f_d x^d) = f_0 x + f_1 x^{[m]} + \cdots + f_d x^{[md]}.
\end{equation}
In particular, $ L(x^n-1) = x^{[mn]} - x $. Moreover, in such case $ F^*(x) = F_*(x) $, $ F^0(x) = F_0(x) $ and $ F^\perp(x) = F^\top(x) $, and coincide by the previous ring isomorphism to the definitions in (\ref{star}), (\ref{zero}) and (\ref{dual}), respectively.

Hence the results in this section give those in the previous one that have to do with divisibility of generator and check polynomials. However, the root description will be essentially different. As we will see, it will not be necessary to assume that $ q $ and $ n $ are coprime, and finding roots of linearized polynomials can always be done efficiently (see \cite[Chapter 3]{lidl}). 
\end{remark}

Before going on, we will establish a result analogous to Proposition \ref{polynomials}. In the linearized case, if $ G(x) $ and $ H(x) $ are coprime on both sides, then there exist an idempotent generator $ E(x) $ of $ C $ by \cite[Theorem 2]{rootskew}, and the linear skew cyclic code with generator and check polynomials $ H(x) $ and $ G(x) $, respectively, is a complementary space of $ C $ by \cite[Proposition 5]{rootskew}. We denote it by $ C^c $. It also has an idempotent generator given by $ x - E(x) $ \cite[Proposition 5]{rootskew}.

\begin{proposition}
The following are equivalent:
\begin{enumerate}
\item
$ C $ is Galois closed.
\item
$ G(x) \in \mathcal{L}_{q^r} \mathbb{F}_q[x] $.
\item
$ H(x) \in \mathcal{L}_{q^r} \mathbb{F}_q[x] $.
\item
$ Z(G) \subseteq \mathbb{F}_{q^{rn}} $ is an ($ \mathbb{F}_{q^r} $-linear) Galois closed space over $ \mathbb{F}_q $. That is, $ Z(G)^q = Z(G) $ (also called $ q $-modulus in \cite[Chapter 3]{lidl}).
\item
(If $ G(x) $ and $ H(x) $ are coprime on both sides) $ E(x) \in \mathbb{F}_q[x] $.
\item
(If $ G(x) $ and $ H(x) $ are coprime on both sides) $ C^c $ is Galois closed.
\end{enumerate}
\end{proposition}
\begin{proof}
Analogous to that of Proposition \ref{polynomials}. \\
\end{proof}

It is proven in \cite{skewcyclic2, gabidulin, qcyclic} that $ H^\perp(x) $ is the generator polynomial of $ C^\perp $. We now find its check polynomial:

\begin{lemma} \label{lin check dual comp}
The check polynomial of $ C^\perp $ is $ G^\top(x) $.
\end{lemma}
\begin{proof}
Let $ \widetilde{G}(x) = \widetilde{G}_0 x + \widetilde{G}_1 x^{[r]} + \cdots + \widetilde{G}_{n-k} x^{[(n-k)r]} $ be the check polynomial of $ C^\perp $. It is shown in \cite{skewcyclic2} that $ C^\perp $ has a parity check matrix of the form
\begin{displaymath}
\left(
\begin{array}{ccccccc}
G_0 & G_1 & \ldots & G_{n-k} & 0 & \ldots & 0 \\
0 & G_0^{[r]} & \ldots & G_{n-k-1}^{[r]} & G_{n-k}^{[r]} & \ldots & 0 \\
\vdots & \vdots & \ddots & \vdots & \vdots & \ddots & \vdots \\
0 & 0 & \ldots & G_0^{[(k-1)r]} & G_1^{[(k-1)r]} & \ldots & G_{n-k}^{[(k-1)r]} \\
\end{array} \right).
\end{displaymath}
By \cite[Theorem 1, items 4 and 6]{rootskew}, there is a unique parity check matrix of that form and hence it holds that $ \widetilde{G}_i^{[(n-k+i)r]} = G_{n-k-i}/G_0 $. Raising this equality to the power $ [(k+i)r] $ we obtain $ \widetilde{G}_i = \widetilde{G}_i^{[nr]} = (G_{n-k-i}/G_0)^{[(k+i)r]} $, for $ i = 0,1,2, \ldots, n-k $, since $ m $ divides $ rn $, and we are done.
\end{proof}

On the other hand, we have the following:

\begin{lemma} \label{perp top}
For a $ q^r $-polynomial $ F(x) = F_0x + F_1 x^{[r]} + \cdots + F_{d} x^{[rd]} \in \mathcal{L}_{q^r} \mathbb{F}_{q^m}[x] $ that divides $ x^{[rn]} - x $, with $ F_d \neq 0 $, it holds that 
$$ F^{\perp \top}(x) = F^{\top \perp}(x) = F(x) / F_d, $$
$$ (F_*)^\perp(x) = (F^\perp)^*(x) \quad \textrm{and} \quad (F_0)^\perp(x) = (F^\perp)^0(x), $$
and analogously replacing $ \perp $ by $ \top $ in the last two equalities.
\end{lemma}
\begin{proof}
The first two equalities are straightforward computations. For the last two equalities, it is enough to observe again that $ \theta_i(F^\perp(x)) = \theta_i(F(x))^\perp $ and use the previous two equalities. Analogously replacing $ \perp $ by $ \top $.
\end{proof}

We will need the following result, which is \cite[Theorem 4]{rootskew}:

\begin{lemma}[\textbf{\cite[Theorem 4]{rootskew}}] \label{theorem 4}
Assume that $ C_1, C_2 \subseteq \mathbb{F}_{q^m}^n $ are linear $ q^r $-cyclic codes with generator polynomials $ G_1(x) $ and $ G_2(x) $, respectively. Then
\begin{enumerate}
\item
$ C_1 \cap C_2 $ is the $ q^r $-cyclic code with generator polynomial $ M(x) = {\rm lcm}(G_1(x), G_2(x)) $ and $ Z(M) = Z(G_1) + Z(G_2) $.
\item
$ C_1 + C_2 $ is the $ q^r $-cyclic code with generator polynomial $ D(x) = {\rm gcd}(G_1(x), G_2(x)) $ and $ Z(D) = Z(G_1) \cap Z(G_2) $.
\end{enumerate}
\end{lemma}

Finally, we may compute the generator and check polynomials of $ C^* $ and $ C^0 $, seen as $ q^r $-cyclic codes:

\begin{proposition} \label{gen skew closure}
The generator and check polynomials of $ C^* $ are $ G^*(x) $ and $ H_0(x) $, respectively, and the generator and check poylnomials of $ C^0 $ are $ G^0(x) $ and $ H_*(x) $, respectively.
\end{proposition}
\begin{proof}
By the previous lemma, if $ C_1, C_2 \subseteq \mathbb{F}_{q^m}^n $ are $ q^r $-cyclic codes with generator polynomials $ G_1(x), G_2(x) $, respectively, and check polynomials $ H_1(x), H_2(x) $, respectively, it holds that $ C_1 + C_2 $ and $ (C_1+C_2)^\perp = C_1^\perp \cap C_2^\perp $ have generator polynomials $ {\rm gcd}(G_1(x),G_2(x)) $ and $ {\rm lcm}(H_1^\perp(x), H_2^\perp(x)) $ (on the right), respectively. By the previous lemma and Lemma \ref{lin check dual comp}, the check polynomial of $ C_1 + C_2 $ is then $ {\rm lcm}(H_1^\perp(x), H_2^\perp(x))^\top $.

We obtain the result for $ C^* $ by applying this iteratedly to $ C, \theta_1(C), \theta_2(C), \ldots, \theta_{m-1}(C) $, observing that the generator and check polynomials of $ \theta_i(C) $ are $ \theta_i(G(x)) $ and $ \theta_i(H(x)) $, respectively,  for $ i = 0,1,2, \ldots, m-1 $. Similarly for $ C^0 $. \\
\end{proof}

We know from Lemma \ref{galois cyclic} that $ C^* $ and $ C^0 $ are skew cyclic of all orders. In the previous sections we used their cyclic (or $ q^0 $-cyclic) nature and their conventional generator and check polynomials. We may relate them with the generator and check polynomials obtained in the previous proposition.

For that purpose, we define the operator $ L : \mathbb{F}_{q^m}[x] \longrightarrow \mathcal{L}_{q^r} \mathbb{F}_{q^m}[x] $ by
\begin{equation}
L(f_0 + f_1x + \cdots + f_d x^d) = f_0 x + f_1 x^{[r]} + \cdots + f_d x^{[rd]},
\end{equation}
which coincides with the map in (\ref{conventional equi to linearized}) for $ r = m $.

\begin{proposition} \label{associate gen}
Let the notation be as in the previous proposition, and let $ g^*(x) $, $ g^0(x) $ be the generator (conventional) polynomials of $ C^* $ and $ C^0 $, respectively, and let $ h^0(x) $ and $ h^*(x) $ be their check (conventional) polynomials, respectively. Then
$$ G^*(x) = L(g^*(x)), \quad H_0(x) = L(h^0(x)), $$
$$ G^0(x) = L(g^0(x)), \quad \textrm{and} \quad H_*(x) = L(h^*(x)). $$
\end{proposition}
\begin{proof}
It follows from the uniqueness of the generator and parity check matrices for cyclic and $ q^r $-cyclic codes given by their generator and check polynomials. See \cite[Theorem 4.2.1 and Theorem 4.2.7]{pless} for the cyclic case, and \cite{skewcyclic2} and \cite[Theorem 1]{rootskew} for the $ q^r $-cyclic case.
\end{proof}

On the other hand, we have the following relations between the root spaces of the generator and check polynomials of $ C $, $ C^* $ and $ C^0 $, as in Proposition \ref{coprime then}.

\begin{proposition}
It holds that
\begin{enumerate}
\item
$ Z(G^*) = \bigcap_{i=0}^{m-1} Z(G)^{[i]} $ and $ Z(G^0) = \sum_{i=0}^{m-1} Z(G)^{[i]} $.
\item
$ \dim_{\mathbb{F}_{q^r}}(Z(H_*)) = \dim_{\mathbb{F}_{q^r}}( \bigcap_{i=0}^{m-1} Z(H^\perp)^{[i]}) $.
\item
$ \dim_{\mathbb{F}_{q^r}}(Z(H_0)) = \dim_{\mathbb{F}_{q^r}}(\sum_{i=0}^{m-1} Z(H^\perp)^{[i]}) $.
\end{enumerate}
\end{proposition}
\begin{proof}
The first item follows from Lemma \ref{theorem 4} and the fact that $ Z(\theta_i(G)) = Z(G)^{[i]} $, for $ i = 0,1,2, \ldots, m-1 $.

On the other hand, since $ H_*(x) $ divides $ x^{[rn]} - x $ on the right, it also divides it conventionally, and hence it has simple roots. Hence it holds that $ \dim_{\mathbb{F}_{q^r}}(Z(H_*)) = \deg_{q^r}(H_*(x)) $ and similarly for $ (H_*)^\perp(x) $. Thus
$$ \dim_{\mathbb{F}_{q^r}}(Z(H_*)) = \deg_{q^r}(H_*(x)) = \deg_{q^r}((H_*)^\perp(x)) = \dim_{\mathbb{F}_{q^r}}(Z((H_*)^\perp)). $$
Again by Lemma \ref{theorem 4} and Lemma \ref{perp top}, we have that
$$ Z((H_*)^\perp) = \bigcap_{i=0}^{m-1} Z(H^\perp)^{[i]}, $$
using again the fact that $ Z(\theta_i(H^\perp)) = Z(H^\perp)^{[i]} $, for $ i = 0,1,2, \ldots, m-1 $. Therefore item 2 follows. Item 3 is proven in a similar way.
\end{proof}

We may now state a similar result to Theorem \ref{rank degenerate cyclic}:

\begin{theorem} \label{rank degenerate skew cyclic}
It holds that
$$ l_R(C) = n - \deg_{q^r}({\rm gcd}(G(x), \theta_1(G(x)), \ldots, \theta_{m-1}(G(x)))) $$
$$ = \deg_{q^r}({\rm lcm}(H^\perp(x), \theta_1(H^\perp(x)), \ldots, \theta_{m-1}(H^\perp(x)))) $$
$$ = n - \dim_{\mathbb{F}_{q^r}} \left( \bigcap_{i=0}^{m-1} Z(G)^{[i]} \right) = \dim_{\mathbb{F}_{q^r}} \left( \sum_{i=0}^{m-1} Z(H^\perp)^{[i]} \right). $$
\end{theorem}
\begin{proof}
The first two equalities follow from Theorem \ref{char different lengths}, item 1, and Proposition \ref{associate gen}. The next two equalities follow from the previous proposition and the fact that $ \dim_{\mathbb{F}_{q^r}}(Z(G^*)) = \deg_{q^r}(G^*(x)) $ and $ \dim_{\mathbb{F}_{q^r}}(Z(H_0)) = \deg_{q^r}(H_0(x)) $, since they have simple roots. 
\end{proof}

Now we obtain the following characterizations of rank degenerate skew cyclic codes:

\begin{corollary} \label{rank degenerate skew}
The following conditions are equivalent:
\begin{enumerate}
\item
$ C $ is rank degenerate. That is, $ l_R(C) < n $.
\item
$ {\rm gcd}(G(x), \theta_1(G(x)), \ldots, \theta_{m-1}(G(x))) \neq x $.
\item
$ {\rm lcm}(H^\perp(x), \theta_1(H^\perp(x)), \ldots, \theta_{m-1}(H^\perp(x))) \neq x^{[rn]} - x $.
\item
$ \bigcap_{i=0}^{m-1} Z(G)^{[i]} \neq \{ \mathbf{0} \} $.
\item
$ \sum_{i=0}^{m-1} Z(H^\perp)^{[i]} \neq \mathbb{F}_{q^{rn}} $.
\item
$ G(x) $ is divisible on the right by some polynomial $ F(x) \in \mathcal{L}_{q^r} \mathbb{F}_{q}[x] $ in $ \mathcal{L}_{q^r} \mathbb{F}_{q^m}[x] $ with $ \deg_{q^r}(F(x)) > 0 $.
\item
$ H(x) $ divides on the right some polynomial $ F(x) \in \mathcal{L}_{q^r} \mathbb{F}_{q}[x] $ in $ \mathcal{L}_{q^r} \mathbb{F}_{q^m}[x] $ with $ \deg_{q^r}(F(x)) < n $.
\end{enumerate} 
\end{corollary}

\section{Attaining the rank length by pseudo-skew cyclic codes}

So far we have tried to find the linear skew cyclic code of smallest length that is rank equivalent to a given one $ C $. We have given upper bounds on that length and seen that a general lower bound is $ l_R(C) $, although it is not clear that this length can be attained by a linear skew cyclic code that is rank equivalent to $ C $.

On the other hand, in this section we will see that the length $ l_R(C) $ is always attained by some linear pseudo-skew cyclic code. As skew cyclic codes, pseudo-skew cyclic codes were introduced in \cite{gabidulin} for $ r=1 $ and $ n=m $, and then independently in \cite{qcyclic} for $ r=1 $ and in \cite{skewcyclic2} for general parameters. They are not invariant by $ q^r $-shifting operators, but have similar conventional-polynomial and linearized-polynomial representations.

We start by defining the well-known pseudo-cyclic codes:

\begin{definition}
Let $ f(x) \in \mathbb{F}_{q^m}[x] $ be of degree $ n $. For a linear code $ C \subseteq \mathbb{F}_{q^m}^n $, we define $ C_{f(x)}(x) $ as the image of $ C $ in $ \mathbb{F}_{q^m}[x] / (f(x)) $ by the linear vector space isomorphism $ \mathbb{F}_{q^m}^n \longrightarrow \mathbb{F}_{q^m}[x] / (f(x)) $ given by
$$ (c_0,c_1, \ldots, c_{n-1}) \mapsto c_0 + c_1x + \cdots + c_{n-1} x^{n-1}. $$
Then we say that $ C $ is pseudo-cyclic if $ C_{f(x)}(x) $ is an ideal in $ \mathbb{F}_{q^m}[x] / (f(x)) $, for some $ f(x) \in \mathbb{F}_{q^m}[x] $ of degree $ n $.
\end{definition}

Fix now a linear cyclic code $ C \subseteq \mathbb{F}_{q^m}^n $, and let the notation be as in Section \ref{conventional}. We have the following:

\begin{theorem} \label{ideal equivalence}
The map $ \phi : \mathbb{F}_{q^m}[x] / (h^0(x)) \longrightarrow (g^*(x)) / (x^n-1) $ given by
$$ \phi(f(x)) = f(x)g^*(x) $$
is well-defined, maps ideals to ideals and constitutes a rank equivalence when seeing its domain and codomain as linear Galois closed spaces.
\end{theorem}
\begin{proof}
First of all, it is well-defined since $ h^0(x) g^*(x) = x^n - 1 $. It is linear since it preserves additions and $ \phi(p(x)f(x)) = p(x) \phi(f(x)) $, for all $ p(x), f(x) \in \mathbb{F}_{q^m}[x] $. For the same reason it maps ideals to ideals.

On the other hand, if $ f(x)g^*(x) = 0 $ in the quotient $ (g^*(x)) / (x^n-1) $, then $ f(x)g^*(x) = p(x)(x^n-1) $ for some polynomial $ p(x) \in \mathbb{F}_{q^m}[x] $, which implies that $ f(x) = p(x)h^0(x) $. Therefore, $ \phi $ is one to one. Since it is obviously onto, we conclude that it is a vector space isomorphism.

Finally, since $ g^*(x) \in \mathbb{F}_q[x] $, we see that $ \phi $ maps polynomials of degree less that $ k $ with coefficients in $ \mathbb{F}_q $ to polynomials with coefficients in $ \mathbb{F}_q $, and hence it is a rank equivalence by Lemma \ref{rank equivalences}.
\end{proof}

Therefore the following consequence follows immediately:

\begin{corollary}
The length $ l_R(C) $ is attained by a linear pseudo-cyclic code that is an ideal in the quotient ring $ \mathbb{F}_{q^m}[x] / (h^0(x)) $. 
\end{corollary}

\begin{remark}
If $ h^0(x) = x^e -1 $, for some positive integer $ e $, then the pseudo-cyclic code in the previous corollary is actually cyclic. This also follows from Remark \ref{attaining}. \\
\end{remark}

In a completely analogous way, we may state similar results for the general skew cyclic case. However, in this case we may only construct quotient rings with two-sided ideals, since the ring of linearized polynomials is not commutative. For that purpose, we consider the the center of $ \mathcal{L}_{q^r} \mathbb{F}_{q^m}[x] $, denoted by $ \mathbb{C}(\mathcal{L}_{q^r} \mathbb{F}_{q^m}[x]) $ and defined as the set of $ q^r $-polynomials over $ \mathbb{F}_{q^m} $ that commute with every other $ q^r $-polynomial over $ \mathbb{F}_{q^m} $. It is well-known that
$$ \mathbb{C}(\mathcal{L}_{q^r} \mathbb{F}_{q^m}[x]) = \mathcal{L}_{q^l} \mathbb{F}_{q^d}[x], $$
where $ l = {\rm lcm}(m,r) $ and $ d = {\rm gcd}(m,r) $. We may now proceed exactly as in the conventional case:

\begin{definition}
Let $ F(x) \in \mathbb{C}(\mathcal{L}_{q^r} \mathbb{F}_{q^m}[x]) $ such that $ \deg_{q^r}(F(x)) = n $. For a linear code $ C \subseteq \mathbb{F}_{q^m}^n $, we define $ C_{F(x)}(x) $ as the image of $ C $ in $ \mathcal{L}_{q^r} \mathbb{F}_{q^m}[x] / (F(x)) $ by the linear vector space isomorphism $ \mathbb{F}_{q^m}^n \longrightarrow \mathcal{L}_{q^r} \mathbb{F}_{q^m}[x] / (F(x)) $ given by
$$ (c_0,c_1, \ldots, c_{n-1}) \mapsto c_0x + c_1x^{[r]} + \cdots + c_{n-1} x^{[(n-1)r]}. $$
Then we say that $ C $ is pseudo-skew cyclic (of order $ r $) if $ C_{F(x)}(x) $ is a left ideal in $ \mathcal{L}_{q^r} \mathbb{F}_{q^m}[x] / (F(x)) $, for some $ F(x) \in \mathbb{C}(\mathcal{L}_{q^r} \mathbb{F}_{q^m}[x]) $ such that $ \deg_{q^r}(F(x)) = n $.
\end{definition}

Fix now a linear $ q^r $-cyclic code $ C \subseteq \mathbb{F}_{q^m}^n $, and let the notation be as in Section \ref{linearized}. We have the following:

\begin{theorem}
Assume that $ H_0(x) $ is central. Then the map $ \phi : \mathcal{L}_{q^r} \mathbb{F}_{q^m}[x] / (H_0(x)) \longrightarrow (G^*(x)) / (x^{[rn]} - x) $ given by
$$ \phi(F(x)) = F(x) \otimes G^*(x) $$
is well-defined, maps left ideals to left ideals and constitutes a rank equivalence when seeing its domain and codomain as linear Galois closed spaces.
\end{theorem}
\begin{proof}
Analogous to that of Theorem \ref{ideal equivalence}, taking into account the non-commutativity of the symbolic product of linearized polynomials.
\end{proof}

Hence the following consequence follows immediately:

\begin{corollary}
If $ H_0(x) $ is central, then the length $ l_R(C) $ is attained by a linear pseudo-skew cyclic code that is a left ideal in the quotient ring $ \mathcal{L}_{q^r} \mathbb{F}_{q^m}[x] / (H_0(x)) $. 
\end{corollary}

\section*{Acknowledgement}

The author wishes to thank Olav Geil and Diego Ruano for fruitful discussions and careful reading of the manuscript. The author also gratefully acknowledges the support from The Danish Council for Independent Research (Grant No. DFF-4002-00367).

%\nocite*{}
%\bibliography{qcyclicbib2}
\bibliographystyle{plain}

\end{document}